\documentclass[conference]{IEEEtran}
\IEEEoverridecommandlockouts
\usepackage{cite}
\usepackage{amsmath,amssymb,amsfonts,amsthm}
\usepackage{algorithmic}
\usepackage{graphicx}
\usepackage{textcomp}
\usepackage{tikz}
\usepackage{xcolor}
\usepackage{hyperref}
\usepackage{relsize}
\newtheorem{theorem}{Theorem}[]
\newtheorem{corollary}{Corollary}[]

\def\BibTeX{{\rm B\kern-.05em{\sc i\kern-.025em b}\kern-.08em
    T\kern-.1667em\lower.7ex\hbox{E}\kern-.125emX}}
\newcommand{\ket}[1]{|#1\rangle}
\newcommand{\bra}[1]{\langle#1|}

\title{\LARGE \bf Complete Positivity of Subsystems in Quantum Dynamics}
\author{Anumita Mukhopadhyay$^{*}$, Praggnyamita Ghosh and Shibdas Roy% <-this % stops a space
\thanks{AM and SR are with the Center for Quantum Engineering, Research and Education (CQuERE), TCG Centres for Research and Education in Science and Technology (TCG CREST), Kolkata 700091, India and Academy of Scientific and Innovative Research (AcSIR), Ghaziabad 201002, India. PG is with Ramakrishna Mission Vivekananda Educational and Research Institute, Belur Math, Howrah 711202, India and Center for Quantum Engineering, Research and Education (CQuERE), TCG Centres for Research and Education in Science and Technology (TCG CREST), Kolkata 700091, India.\linebreak
$^{*}$Email: {\tt\small anumitamukherjee455@gmail.com}}%
}

\begin{document}

\maketitle
\begin{abstract}
Although many quantum channels satisfy Completely Positive Trace Preserving (CPTP) condition, there are valid quantum channels that can be non-completely positive (NCP). 
%In a search of the conditions of noisy evolution to be a useful resource for quantum computing, 
%we study the relation of complete positivity (CP) with unitality, where we find that a map must be non-unital in order to be NCP, but not vice-versa. 
As memory effects can provide advantages in the dynamics of noisy quantum systems, we investigate the relative CP condition and the CP-divisibility condition of the system and environment subsystems of a joint system-environment quantum state evolving noiselessly. We show that the system and environment channels must be both CP (NCP) or CP-divisible (CP-indivisible) for the evolution in the joint system-environment space to be unitary. We illustrate our results with examples of Bell state created from $\ket{00}$, GHZ state created from $\ket{000}$, W state created from $\ket{100}$, and the partial transpose (PT) operation acting on the Bell state.
\end{abstract}


\section{Introduction}
Open quantum systems, comprising a system Hilbert space, when in interaction with an environment Hilbert space, evolve jointly via unitary transformations. This interaction between the system and environment spaces introduces noisy evolution of the system. The dynamics of an open quantum system is characterized by using operator-sum representation, also known as the Choi-Kraus representation. %According to operator-sum representation, quantum channels are Completely Positive Trace Preserving (CPTP) maps. Among them, unitary maps preserve norm, purity as well as entropy and allow no loss of informations. All quantum gates are CPTP maps with such special properties as they are unitaries. Apart from that, measurement is also a CPTP map. However,  noisy CPTP channels introduce decoherence in the system. 
The operator-sum representation or Kraus representation for noisy CPTP channels are often a convex mixture of unitaries. Such channels are known as unital quantum channels. %They also preserve maximally mixed states, which correspond to the origin of the Bloch sphere. In contrast, when the Kraus operators cannot be written as convex combination of unitaries and fails to preserve the maximally mixed state, the channel is called a non-unital quantum channel. 
According to Choi-Kraus theorem, Kraus operators $A_k$ satisfy the relation, $\sum _k A_k^\dagger A_k \leq \mathbb{I}$, where equality ensures trace preservation. Moreover, Kraus operators for a unital channel satisfy the relation $\sum _k A_k A_k ^\dagger = \mathbb{I}$, while those of a non-unital channel do not satisfy this relation.

Quantum computing currently faces a significant hurdle due to the presence of noise, which usually degrades computational efficiency by introducing decoherence. Error correction and noise mitigation procedures often involve cost-inefficient techniques. Consequently, exploring the usefulness of noise rather than avoiding it, presents a promising avenue \cite{urbasi}. This can potentially evade the limitations in Noisy Intermediate Scale Quantum (NISQ) Computing. %Several lines of research already show how the noise can be exploited and become a valuable resource. For instance, 
Ref.~\cite{cirac} highlights how dissipation can be an alternative approach to quantum computing and state engineering of wide range of highly correlated states, even without coherent dynamics. Similarly, Ref.~\cite{temporal} shows that amplitude damping noise can be effectively harnessed for Quantum Reservoir Computing (QRC) and information processing. This concept of engineered dissipation for quantum information processing is a promising field, as further established in Ref.~\cite{dissipaton}. Interestingly, noise can play a role in generating quantum correlations. Ref.~\cite{qcorr} indicates that non-unital noisy channels can foster quantum correlations in multi-qubit systems, while unital noisy channels can do the same for multi-qudit systems. Furthermore, research like Ref.~\cite{req} illustrates that mixed entangled states can exhibit significantly more non-classicality than separable and pure entangled states. In fact, noise can robustly enhance entanglement within a quantum system, as shown in Ref.~\cite{noise}. Besides, it is shown in Ref.~\cite{nonmarkov} that memory effects can possibly revive quantum correlations after some initial decay. The backflow of information from environment to the system has different thought-provoking features as shown in Ref.~\cite{banerjee,BLP,Utagi,Brodutch}. Similar to non-unitality, non-Markovianity is shown to be a useful resource in Ref.~\cite{srikanth,sibasish,Kumar,pathak,open_non_markov,teleport,naikoo}. Thus, in order to summarize the conditions on the noisy quantum channels to be useful, we need to investigate non-unitality as well as non-Markovianity.
In our previous work \cite{self}, we have already shown that for unitary evolution of a joint system-environment quantum state, if the system evolves unitally (non-unitally), then the environment will also evolve unitally (non-unitally). In extension to this, here we explore complete positivity conditions on the system and environment channels.
%To narrow down to the conditions on the noisy quantum channels to be useful for quantum computing, it is necessary to investigate the intricate relation between unitality condition of the quantum channels acting on the system and environment with Complete Positivity (CP) of the channels. 
To investigate further on the non-Markovianity conditions on the system and environment, we need to find out if the system and environment channels are both Complete Positivity (CP)-divisible or not. It is known that non-Markovianity of quantum channels is defined in terms of the presence of intermediate NCP maps. However, there are non-Markovian dynamics despite the absence of CP-indivisibility \cite{Rivas_2014,rivas,Sbannerjee}. Here, we investigate CP-divisibility of system and environment channels to more precisely arrive at the conditions under which memory effects manifest within these quantum processes. Our results are illustrated with Bell state created from $\ket{00}$, GHZ state created from $\ket{000}$, W state created from $\ket{100}$, by splitting the corresponding unitaries into two or three unitaries in series, and the non-CP partial transpose operation(s) acting on the maximally entangled Bell state.

%In Section \ref{sec:CKS}, we discuss Choi-Kraus representation of open quantum channels. In Section \ref{sec:result}, we present our results on CP conditions. In Section \ref{sec:example}, we give examples to testify our results on CP-divisibility conditions of system and environment.

\section{Choi-Kraus representation} \label{sec:CKS}
For $\varepsilon$ to be a valid quantum operation on an initial state $\rho$, the probability of the process to occur should be $0 \leq {\rm Tr}[\varepsilon(\rho)]\leq 1$; on the set of density matrices, it should be convex-linear map given by $\varepsilon(\sum_ip_i \rho_i)=\sum_i p_i \varepsilon(\rho_i)$, and it should be a completely positive map. The following theorem states the afore-mentioned notion of quantum operations \cite{choi}. 
\begin{theorem}
    A map $\varepsilon$ is a valid quantum operation, iff
    \begin{equation}
        \varepsilon(\rho)=\sum_k A_k \rho A_k ^\dagger, 
    \end{equation}
 for a set of operators $\{A_k\}$, such that
\begin{equation*} \sum _k A_k^\dagger A_k \leq \mathbb{I}.
     \end{equation*}
\end{theorem}
The proof can be found in Ref.~\cite{NC}.

If some quantum operation $\varepsilon$ is a trace preserving map, then 
\begin{eqnarray*}
    {\rm Tr}(\varepsilon(\rho))&=& \sum_k {\rm Tr}(A_k \rho A_k ^\dagger)\\
    &=& \sum_k {\rm Tr}(A_k ^\dagger A_k \rho)\\
    &=& {\rm Tr}(\rho)=1,
\end{eqnarray*}
which implies that we must have $\sum_k A_k ^\dagger A_k = \mathbb{I}$, that is known as completeness relation.

If a state $\ket{\Psi}$ is considered on $\mathcal{H_S}\otimes\mathcal{H_R}$, where $\mathcal{H_S}$ is the system Hilbert space and $\mathcal{H_R}$ is a reference Hilbert space, and a positive operator $\Omega$ is taken into account, then
\begin{eqnarray*}
    \bra{\Psi}(\mathbb{I}\otimes\varepsilon)\Omega\ket{\Psi}&=&\sum_k \bra{\Psi}(\mathbb{I}_R\otimes A_k)\Omega (\mathbb{I}_R \otimes A_k ^\dagger)\ket{\Psi}\\
    &=& \sum_k \bra{\Phi_k}\Omega\ket{\Phi_k} \geq 0,
\end{eqnarray*}
where $\ket{\Phi_k}= \mathbb{I}_R\otimes A_k ^\dagger \ket{\Psi}$. As $\Omega \geq 0$, we also have $(\mathbb{I}\otimes \varepsilon) \Omega \geq 0 \implies (\mathbb{I}\otimes \varepsilon) \geq 0 $. Thus, $\varepsilon$ is a CP map, and any CP map can be written as $\varepsilon(\rho)= \sum_k A_k \rho A_k^\dagger$. 

Thus, we can infer from the above theorem that \textbf{the completeness relation can be a witness of complete positivity for a trace-preserving map.} Thus, a trace-preserving (TP) non-completely positive (NCP) map will not satisfy this completeness relation.% In light of such inference, we move forward towards our results.

\section{Results}\label{sec:result}
It is already known that a completely positive trace-preserving (CPTP) map is a valid quantum operation \cite{sudarshan}. However, sometimes NCP maps can also be legitimate description of open quantum systems as shown in Ref.~\cite{linta}. For NCP maps to be a valid quantum channel, they also need to satisfy a trace preserving condition. The Kraus representation for an NCP map is given by \cite{Jagadish_2018}: $D_\alpha=\sqrt{\lambda_\alpha}~{\rm mat} \ket{\Lambda^{(\alpha)}}$ for positive eigenvalues $\lambda_\alpha$ of the $\mathcal{B}$ matrix with eigenvectors $\ket{\Lambda^{(\alpha)}}$ \cite{B-matrix} and $F_\alpha=|\sqrt{\lambda_\alpha}|~{\rm mat} \ket{\Lambda^{(\alpha)}}$ for the negative eigenvalues of the same. Here, ${\rm mat}\ket{\Lambda^{(\alpha)}}$~\emph{matricizes} the vectors $\ket{\Lambda^{(\alpha)}}$; please see Ref.~\cite{B-matrix} for details. The trace preserving condition for the Kraus operators of the NCP map will then be: $\sum_{\alpha=0}^{k-1} D_\alpha^\dagger D_\alpha - \sum_{\alpha=k}^ {N^2-1}F_\alpha^\dagger F_\alpha=\mathbb{I}$, where $\alpha=0,1,2, \cdots k-1$ are for positive eigenvalues among a total number of $N^2$ non-zero eigenvalues.  
We are keen to explore the CP and NCP conditions on the quantum channels acting on the system and environment individually of a joint system-environment quantum state evolving via some unitary.% For an open system dynamics here, we first see CP condition from the perspective of unitality condition of the quantum channel, followed by an inspection of the behaviour of the quantum channel under unitary freedom of Kraus operators. Next, we delve deeper into the CP-condition of system and environment in interaction with each other and CP-divisibility condition as well. 

\begin{corollary}
A CP quantum channel cannot become NCP or vice-versa under the action of a unitary.
\end{corollary}
\begin{proof}
Let a set of Kraus operators $K_i$ be related to another set of Kraus operators $L_j$ via unitary freedom of Kraus operators, given by $L_j=U_{ji}K_i$. Thus,
\begin{eqnarray*}
    \sum_j L_j ^\dagger L_j &=& \sum_{i,j} (U_{ji}K_i)^\dagger (U_{ji}K_i)\\
    &=& \sum_{i,j} K_i ^\dagger U_{ji}^\dagger U_{ji}K_i\\
    &=& \sum_i K_i ^\dagger K_i,
\end{eqnarray*}
since $\sum_j U_{ji}^\dagger U_{ji}=1\, \forall i$ as $U^\dagger U= \mathbb{I}$. So, if $K_i$ is a CPTP map, i.e.~$\sum_i K_i ^\dagger K_i=\mathbb{I}$, then, $L_j$ must also be a CPTP map.% Thus, unitary freedom of Kraus operators cannot change the CP condition of a map.
\end{proof}

We next show that \textbf{if a trace-preserving system channel is CP, then the corresponding trace-preserving environment channel must also be CP}.

\begin{theorem}
    If a trace-preserving channel acting on the system is completely positive, then the trace preserving channel acting on the environment must also be completely positive.% Similarly, if the quantum channel acting on the system is not completely positive, the quantum channel acting on the environment will also be not completely positive.
\end{theorem}
\begin{proof}
    Let us take a state $\ket{\Psi}$ on the Hilbert space $\mathcal{H_S}\otimes \mathcal{H_E}\otimes \mathcal{H_R}$ where $\mathcal{H_S}$ is the system Hilbert space, $\mathcal{H_E}$ is the environment Hilbert space and $\mathcal{H_R}$ is a reference Hilbert space.
    Here, we take $\varepsilon_S$ as the quantum channel acting on the system and $\varepsilon_E$ as the quantum channel acting on the environment. We take another positive operator $\Omega$.
    Now, we can write:
    \begin{eqnarray}
        &&\bra{\Psi}(\mathbb{I}_R \otimes \varepsilon_S \otimes\varepsilon_E)\Omega\ket{\Psi} \\ \nonumber
        &=& \bra{\Psi}\sum_{k,l}(\mathbb{I}_R \otimes A_k \otimes B_l)\Omega (\mathbb{I}_R \otimes A_k ^\dagger \otimes B_l ^\dagger)\ket{\Psi} \\ \nonumber
        &=& \sum_{k,l}\bra{\psi_{kl}}\Omega\ket{\psi_{kl}} \geq 0,
    \end{eqnarray}
since $\Omega$ is a positive operator. Here $\ket{\psi_{kl}}=(\mathbb{I}_R \otimes A_k ^\dagger \otimes B_l ^\dagger)\ket{\Psi}$.

Hence, 
\begin{align*}
    &&(\mathbb{I}_R \otimes \varepsilon_S\otimes \varepsilon_E)\Omega &\geq& 0 \\ 
    &\implies& (\mathbb{I}_R \otimes \varepsilon_S\otimes \varepsilon_E) &\geq& 0 \\
   &\implies& \varepsilon_S\otimes \varepsilon_E &\geq& 0.
\end{align*}   
Thus, $\varepsilon_S$ and $\varepsilon_E$ should be both completely positive or both non completely positive.

Further, let us define a unitary, $U$, acting on the system-environment joint quantum state. Let us define the state of the system as $\ket{\Psi}=\sum_j \sqrt{q_j}\ket{\psi_j}$, undergoing a map $\mathbb{A}$ with noise operators $A_i$, and the state of the environment as $\ket{\Lambda}=\sum_i\sqrt{p_i}\ket{a_i}$, undergoing a map $\mathbb{B}$ with noise operators $B_j$. Then the unitary acts on the joint quantum state as follows:
\begin{eqnarray}\label{eq:unitary}
    U\ket{\Psi}\ket{\Lambda}&=& \sum_i A_i \ket{\Psi}\ket{a_i}=\sum_j B_j \ket{\psi_j}\ket{\Lambda}\\ \nonumber
    \implies U^2\ket{\Psi}\ket{\Lambda} &=& \sum_{i,j}(A_i \otimes B_j) \ket{\psi_j}\ket{a_i}
\end{eqnarray}
If the system map $\mathbb{A}$ is CP, then we will have $A_i=\bra{a_i}U\ket{\Lambda}$. Then, it follows from the above equation that:
\begin{eqnarray}\label{eq:firstkraus}
     (A_i \otimes B_j) &=& \nonumber\bra{a_i}\bra{\psi_j}U^2\ket{\Psi}\ket{\Lambda}\\ &=& \bra{a_i}U\ket{\Lambda}\otimes \bra{\psi_j}U\ket{\Psi}\\ \nonumber
    \implies \sum_{i,j} A_i ^\dagger A_i \otimes  B_j ^\dagger B_j &=& \sum_{i,j} \bra{\Lambda}U^\dagger\ket{a_i}\bra{a_i}U\ket{\Lambda} \\ \nonumber
    &&\otimes \bra{\Psi}U^\dagger \ket{\psi_j}\bra{\psi_j}U\ket{\Psi} \\ \nonumber
    &=& \bra{\Lambda}U^\dagger U \ket{\Lambda} 
    \otimes \bra{\Psi}U^\dagger U\ket{\Psi}  \\ 
    &=&\mathbb{I}
\end{eqnarray}
Hence, if $\sum_i A_i ^\dagger A_i=\mathbb{I}$, we must have $\sum_j B_j ^\dagger B_j= \mathbb{I}$, i.e.~if $\mathbb{A}$ is a CPTP map, $\mathbb{B}$ will also be a CPTP map. Likewise, if $\mathbb{A}$ is non-CP, then $\mathbb{B}$ must also be non-CP.
\end{proof}

A CPTP map $\varepsilon(t_2,t_0)$ will be defined as a CP divisible \cite{Rivas_2014,Utagi} map if for an intermediate time step $t_1$, we have:
\begin{equation}\label{eq:CP_divisibility}
    \varepsilon(t_2,t_0)=\varepsilon(t_2,t_1)\varepsilon(t_1,t_0),
\end{equation}
such that $\varepsilon(t_2,t_1)$ and $\varepsilon(t_1,t_0)$ are both CP maps for $t_0 \leq t_1 \leq t_2$. It is known that CP-indivisibility implies memory effects, but CP-divisible maps can also have memory effects \cite{Utagi}. Here, we explore CP-divisibility conditions of system and environment, evolving unitarily as a joint quantum state.  
\begin{theorem}
If the quantum channel acting on the system is CP-divisible, then the quantum channel acting on the environment must also be CP-divisible.% Similarly, if the quantum channel acting on the system is CP-indivisible, the quantum channel acting on the environment must also be CP-indivisible.    
\end{theorem}
\begin{proof}
Let us define the state of the system as $\ket{\Psi}=\sum_j \sqrt{q_j}\ket{\psi_j}$ and the state of the environment as $\ket{\Lambda}=\sum_i\sqrt{p_i}\ket{a_i}$. 
According to (\ref{eq:CP_divisibility}), let $U$ be $U=W \cdot V$, where $V$ consists of a system map $\mathbb{C}$ with noise operators $C_k$ and an environment map $\mathbb{D}$ with noise operators $D_l$ and $W$ consists of a system map $\mathbb{E}$ with noise operators $E_m$ and an environment map $\mathbb{F}$ with noise operators $F_n$. Following (\ref{eq:unitary}), we can write:
\begin{eqnarray*}
    V\ket{\Psi}\ket{\Lambda} = \sum_k C_k\ket{\Psi}\ket{a_k}=\sum_l D_l \ket{\psi_l}\ket{\Lambda}.
\end{eqnarray*}
From (\ref{eq:unitary}), we have $U^2\ket{\Psi}\ket{\Lambda} = \sum_{i,j}(A_i \otimes B_j) \ket{\psi_j}\ket{a_i}$. Similarly, for $V$,
\begin{eqnarray}\label{eq:V}
    V^2\ket{\Psi}\ket{\Lambda}= \sum_{k,l} (C_k \otimes D_l) \ket{\psi_l}\ket{a_k},
\end{eqnarray}
and 
\begin{eqnarray}\label{eq:W}
    W^2 \ket{\Phi}\ket{\Gamma} = \sum_{m,n} (E_m \otimes F_n) \ket{\phi_n}\ket{b_m},
\end{eqnarray}
where $\ket{\Phi}=\sum_n \sqrt{r_n}\ket{\phi_n}$ and $\ket{\Gamma}=\sum_m \sqrt{s_m}\ket{b_m}$, and we have $\ket{\Phi}\ket{\Gamma}=V\ket{\Psi}\ket{\Lambda}$.

Thus, from (\ref{eq:V}) and (\ref{eq:W}), we can write:
\begin{eqnarray*}
    W^2 \cdot V^2 \ket{\Psi}\ket{\Lambda} &=& \sum_{k,l,m,n} (E_m C_k \otimes F_n D_l)\ket{\psi_l}\ket{a_k}\\
    &=& \sum_{i,j}(A_i\otimes B_j) \ket{\psi_j}\ket{a_i}. 
\end{eqnarray*}
The above equation and (\ref{eq:firstkraus}) implies:
\begin{eqnarray*}
    \sum_{i,j} (A_i ^\dagger A_i\otimes B_j ^\dagger B_j) &=& \sum_{k,l,m,n}(C_k ^\dagger E_m^\dagger E_m C_k \otimes D_l ^\dagger F_n ^\dagger F_n D_l ) \\
    &=& \mathbb{I}.
\end{eqnarray*}
Now, let $\mathbb{E}$ be a CP map. Then, we have $\sum_m E_m^\dagger E_m = \mathbb{I}$, which, in turn, implies that we must have $\sum_n F_n ^\dagger F_n=\mathbb{I}$. Then, the above equation yields
$\sum_{k,l} (C_k ^\dagger C_k \otimes D_l ^\dagger D_l)= \mathbb{I}$. Clearly, this means, if $\sum_k C_k ^\dagger C_k=\mathbb{I}$, then we must have $\sum_l D_l ^\dagger D_l = \mathbb{I}$. Hence, we can infer that if the system channel is CP-divisible then the environment channel will also be CP-divisible.
This proof can be generalized in a straightforward manner. Let us say, the quantum state for joint system and environment is given by $\rho_{SE}$. It is evolving as:
\begin{equation*}
    U \rho_{SE}U^\dagger = \sum_i A_i \rho_S A_i ^\dagger \otimes \ket{a_i}\bra{a_i}=\sum_j \ket{\psi_j}\bra{\psi_j}\otimes B_j \rho_E B_j ^\dagger ,
\end{equation*}
where $\rho_S = {\rm Tr}_E (\rho_{SE})=\sum_j q_j \ket{\psi_j}\bra{\psi_j}$ and $\rho_E={\rm Tr}_S (\rho_{SE})=\sum_i p_i \ket{a_i}\bra{a_i}$.
Now, we can write:
\begin{eqnarray*}
    U^2 \rho_{SE}(U^\dagger)^2 = \sum_{i,j} (A_i \otimes B_j) \ket{\psi_j}\bra{\psi_j} \otimes \ket{a_i}\bra{a_i} (A_i ^\dagger \otimes B_j ^\dagger).
\end{eqnarray*}
Let us split our unitary in two parts: $U=W \cdot V$. Thus, we have:
\begin{equation*}
    V^2 \rho_{SE}(V^\dagger)^2 = \sum_{k,l} (C_k \otimes D_l) \ket{\psi_l}\bra{\psi_l}\otimes \ket{a_k}\bra{a_k}(C_k ^\dagger \otimes D_l ^\dagger),
\end{equation*}
and 
\begin{equation*}
    W^2 \sigma_{SE} (W^\dagger)^2= \sum_{m,n} (E_m\otimes F_n) \ket{\phi_n}\bra{\phi_n}\otimes \ket{b_m}\bra{b_m}(E_m ^\dagger \otimes F_n ^\dagger),
\end{equation*}
where $\sigma_{SE}= V \rho_{SE}V^\dagger$ and ${\rm Tr}_E(\sigma_{SE})= \sum_n s_n \ket{\phi_n}\bra{\phi_n}= \sum_k C_k \rho_S C_k ^\dagger$, ${\rm Tr}_S(\sigma_{SE})= \sum_m r_m \ket{b_m}\bra{b_m}=\sum_l D_l \rho_E D_l ^\dagger$. Then, we get:
\begin{eqnarray*}
    U^2 \rho_{SE}(U^\dagger)^2 &=&W^2 V^2 \rho_{SE}(V^\dagger)^2 (W^\dagger)^2 \\
    &=& \sum_{i,j} (A_i \otimes B_j) \ket{\psi_j}\bra{\psi_j} \otimes \ket{a_i}\bra{a_i} (A_i ^\dagger \otimes B_j ^\dagger)\\
    &=& \sum_{k,l,m,n} (E_m C_k \otimes F_n D_l)\ket{\psi_l}\bra{\psi_l}\otimes\\
    &&\ket{a_k}\bra{a_k}(C_k ^\dagger E_m ^\dagger \otimes D_l ^\dagger F_n ^\dagger).
\end{eqnarray*}
Thus, we will have: $\sum_{i,j} A_i ^\dagger A_i \otimes B_j ^\dagger B_j = \sum_{k,l,m,n}(C_k ^\dagger E_m ^\dagger E_m C_k \otimes D_l ^\dagger F_n ^\dagger F_n D_l)= \mathbb{I}$. Let $\sum_m E_m ^\dagger E_m=\mathbb{I}$, i.e.~the map $\mathbb{E}$ is CP. Then, from our previous result, we must have $\sum_n  F_n ^\dagger F_n =\mathbb{I}$, i.e.~the map $\mathbb{F}$ must also be CP. This implies that we must have:
\begin{equation*}
    \sum_{k,l}C_k ^\dagger C_k \otimes D_l ^\dagger D_l = \mathbb{I}. 
\end{equation*}
Thus, if $\sum_k C_k ^\dagger C_k =\mathbb{I}$, we must have $\sum_l D_l ^\dagger D_l=\mathbb{I}$. That is, if the map $\mathbb{C}$ is CP, then so is the map $\mathbb{D}$. This means that if the system is CP-divisible, then the environment must also be CP-divisible, and if the system is non-CP divisible, then the environment must also be  non-CP divisible.
\end{proof}

\section{Examples}\label{sec:example}
\begin{enumerate} 
    \item Consider a 2-qubit \textbf{Bell state} 
    $\frac{1}{\sqrt{2}}(\ket{00} + \ket{11})$, created from $\ket{00}$, using the unitary:
    \[
    \begin{aligned}
    U &= \text{CNOT}(H \otimes \mathbb{I})\\
    &= \frac{1}{\sqrt{2}} \big[
    \ket{00}\bra{00} + \ket{00}\bra{10} + \ket{01}\bra{01} + \ket{01}\bra{11}\\ 
    &+ \ket{10}\bra{01} + \ket{11}\bra{00} - \ket{11}\bra{10} - \ket{10}\bra{11}
    \big].
    \end{aligned}
    \]
    Let $U = U_2 \cdot U_1$, where $U_1 = H \otimes \mathbb{I}$ and $U_2 = \text{CNOT}$. The input state for $U_1$ is $\ket{00}$. We have:
    \[
    \begin{aligned}
    U_1 &=  \frac{1}{\sqrt{2}}[\ket{00}\bra{00} + \ket{00}\bra{10} + \ket{01}\bra{01} + \ket{01}\bra{11}\\
    &+ \ket{10}\bra{00} -\ket{10}\bra{10} + \ket{11}\bra{01} - \ket{11}\bra{11}].
    \end{aligned}
    \]
    The Kraus operators of the noise acting on the system, i.e., qubit 1 are:
    \[
    S_0 = {}_2\bra{0} U_1 \ket{0}_2 
    =  \frac{1}{\sqrt{2}}[\ket{0}\bra{0} + \ket{0}\bra{1} + \ket{1}\bra{0} - \ket{1}\bra{1}],
    \]
    \[
    S_1 = {}_2\bra{1} U_1 \ket{0}_2 
    = 0,
    \]
    and those of the noise acting on the environment are:
    \[
    E_0 = {}_{1}\bra{0} U_1 \ket{0}_{1} = \frac{1}{\sqrt{2}}[\ket{0}\bra{0} + \ket{1}\bra{1}],
    \]
    \[
    E_1 = {}_{1}\bra{1} U_1 \ket{0}_{1} = \frac{1}{\sqrt{2}}[\ket{0}\bra{0} + \ket{1}\bra{1}].
    \]
    So, we have
    \[
    S_0^{\dagger}S_0 + S_1^{\dagger}S_1
    = \ket{0}\bra{0} +\ket{1}\bra{1} = \mathbb{I},
    \]
    \[
     E_0^{\dagger}E_0 +  E_1^{\dagger}E_1 
    = \ket{0}\bra{0} + \ket{1}\bra{1} = \mathbb{I}.
    \]
    
    Next, the input state for $U_2$ is ${\ket{+0}}$, where
    \[
    \begin{aligned}
    \ket{+} &= \frac{1}{\sqrt{2}}[\ket{0}+\ket{1}],\ket{-} = \frac{1}{\sqrt{2}}[\ket{0}-\ket{1}],
    \end{aligned}
    \]
    We have:
    \[
    \begin{aligned}
    U_2 &=\ket{00}\bra{00} + \ket{01}\bra{01} + \ket{10}\bra{11} + \ket{11}\bra{10} 
    \end{aligned}
    \]
    The Kraus operators of the noise acting on the system, i.e., qubit 1 are:
    \[
    S_0= {}_2\bra{0} U_2 \ket{0}_2  =\ket{0}\bra{0},
    \]
    \[
    S_1 = {}_2\bra{1} U_2 \ket{0}_2 = \ket{1}\bra{1},
    \]
    and those of the noise acting on the environment are:
    \[
    E_0 = {}_{1}\bra{+} U_2 \ket{+}_{1} =0.5[\ket{0}\bra{0} + \ket{0}\bra{1} + \ket{1}\bra{0} + \ket{1}\bra{1}],
    \]
    \[
    E_1 ={}_{1}\bra{-} U_2 \ket{+}_{1} = 0.5[\ket{0}\bra{0} - \ket{0}\bra{1} - \ket{1}\bra{0} + \ket{1}\bra{1}].
    \]
    So, we have
    \[
    S_0^{\dagger}S_0 + S_1^{\dagger}S_1
    = \ket{0}\bra{0} +\ket{1}\bra{1} = \mathbb{I},
    \]
    \[
     E_0^{\dagger}E_0 +  E_1^{\dagger}E_1 
    = \ket{0}\bra{0} + \ket{1}\bra{1} = \mathbb{I}.
    \]
    This implies that both system $S$ and environment $E$ are \textbf{CP-divisible} for Bell state, and they are unital \cite{self}.
    
    \item Consider a 3-qubit \textbf{GHZ state} 
    $\frac{1}{\sqrt{2}}(\ket{000} + \ket{111})$, created from $\ket{000}$, using the unitary:
    \[
    \begin{aligned}
    U &= (\mathbb{I} \otimes \text{CNOT})(\text{CNOT} \otimes \mathbb{I})(H \otimes \mathbb{I} \otimes \mathbb{I})\\
    &= \frac{1}{\sqrt{2}} \big[
    \ket{000}\bra{000} + \ket{000}\bra{100}  +\ket{001}\bra{001} +\ket{001}\bra{101}\\
    &+ \ket{010}\bra{011} + \ket{010}\bra{111}
    + \ket{011}\bra{010} +\ket{011}\bra{110}\\
    &+ \ket{100}\bra{010} -\ket{100}\bra{110} + \ket{101}\bra{011} -\ket{101}\bra{111}\\
    &+ \ket{110}\bra{001} - \ket{110}\bra{101} + \ket{111}\bra{000}-\ket{111}\bra{100}
    \big].
    \end{aligned}
    \]
    Let $U = U_3\cdot U_2 \cdot U_1$, where $U_1 = H \otimes \mathbb{I} \otimes \mathbb{I}$, $U_2 = \text{CNOT} \otimes \mathbb{I}$ and $U_3 = \mathbb{I} \otimes \text{CNOT}$.   
    The input state for $U_1$ is $\ket{000}$. We have:
    \[
    \begin{aligned}
    U_1 &= \frac{1}{\sqrt{2}}[\ket{000}\bra{000} + \ket{001}\bra{001} + \ket{010}\bra{010} + \ket{011}\bra{011} \\
    &-\ket{100}\bra{100} - \ket{101}\bra{101} - \ket{110}\bra{110} - \ket{111}\bra{111}\\
    &+\ket{000}\bra{100} + \ket{001}\bra{101} + \ket{010}\bra{110} + \ket{011}\bra{111}\\
    &+\ket{100}\bra{000} + \ket{101}\bra{001} + \ket{110}\bra{010} + \ket{111}\bra{011}].
    \end{aligned}
    \]
    The Kraus operators of the noise acting on the system, i.e., qubits 1,2 are:
    \[
    \begin{aligned}
    S_0 &= {}_3\bra{0} U_1 \ket{0}_3\\
    &= \frac{1}{\sqrt{2}}[\ket{00}\bra{00} + \ket{00}\bra{10} + \ket{01}\bra{01} + \ket{01}\bra{11} +\\
    &\quad \ket{10}\bra{00} - \ket{10}\bra{10} + \ket{11}\bra{01} - \ket{11}\bra{11}],
    \end{aligned}
    \]
    \[
    S_1 = {}_3\bra{1} U_1 \ket{0}_3 
    = 0,
    \]
    and those of the noise acting on the environment are:
    \[
    E_0 = {}_{12}\bra{00} U_1 \ket{00}_{12} = \frac{1}{\sqrt{2}}[\ket{0}\bra{0} + \ket{1}\bra{1}],
    \]
    \[
    E_1 ={}_{12}\bra{01} U_1 \ket{00}_{12} = 0,
    \]
    \[
    E_2 = {}_{12}\bra{10} U_1 \ket{00}_{12} = \frac{1}{\sqrt{2}}[\ket{0}\bra{0} + \ket{1}\bra{1}],
    \]
    \[
    E_3 = {}_{12}\bra{11} U_1 \ket{00}_{12} = 0.
    \]
    So, we have
    \[
    S_0^{\dagger}S_0 + S_1^{\dagger}S_1
    = \ket{00}\bra{00} +\ket{01}\bra{01} +\ket{10}\bra{10} +\ket{11}\bra{11} = \mathbb{I},
    \]
    \[
     E_0^{\dagger}E_0 +  E_1^{\dagger}E_1 + E_2^{\dagger}E_2 +  E_3^{\dagger}E_3 
    = \ket{0}\bra{0} + \ket{1}\bra{1} = \mathbb{I}.
    \]
    
    Next, the input state for $U_2$ is ${\ket{+00}}$, where
    \[
    \begin{aligned}
    \ket{+} &= \frac{1}{\sqrt{2}}[\ket{0}+\ket{1}],\ket{-} = \frac{1}{\sqrt{2}}[\ket{0}-\ket{1}].
    \end{aligned}
    \]
    We have:
    \[
    \begin{aligned}
    U_2 &= \ket{000}\bra{000} + \ket{001}\bra{001} + \ket{010}\bra{010} + \ket{011}\bra{011}\\
    &+\ket{100}\bra{110} + \ket{101}\bra{111} + \ket{110}\bra{100} + \ket{111}\bra{101}.
    \end{aligned}
    \]
    The Kraus operators of the noise acting on the system, i.e., qubits 1,2 are:
    \[
    S_0 = {}_3\bra{0} U_2 \ket{0}_3 
    = \ket{00}\bra{00} + \ket{01}\bra{01} + \ket{10}\bra{11} + \ket{11}\bra{10},
    \]
    \[
    S_1 = {}_3\bra{1} U_2 \ket{0}_3 
    = 0,
    \]
    and those of the noise acting on the environment are:
    \[
    E_0 = {}_{12}\bra{+0} U_2 \ket{+0}_{12} = 0.5[\ket{0}\bra{0} + \ket{1}\bra{1}],
    \]
    \[
    E_1 = {}_{12}\bra{+1} U_2 \ket{+0}_{12} = 0.5[\ket{0}\bra{0} + \ket{1}\bra{1}],
    \]
    \[
    E_2 = {}_{12}\bra{-0} U_2 \ket{+0}_{12} = 0,
    \]
    \[
    E_3 = {}_{12}\bra{-1} U_2 \ket{+0}_{12} = \frac{1}{\sqrt{2}}[\ket{0}\bra{0} + \ket{1}\bra{1}].
    \]
    So, we have
    \[
    S_0^{\dagger}S_0 + S_1^{\dagger}S_1
    = \ket{00}\bra{00} + \ket{01}\bra{01} + \ket{10}\bra{10} + \ket{11}\bra{11} = \mathbb{I},
    \]
    \[
     E_0^{\dagger}E_0 +  E_1^{\dagger}E_1 + E_2^{\dagger}E_2 +  E_3^{\dagger}E_3 
    = \ket{0}\bra{0} + \ket{1}\bra{1} = \mathbb{I}.
    \]
    Next, the input state for $U_3$ is ${\frac{1}{\sqrt{2}}[\ket{000}+\ket{110}]}$ = $\ket{\phi^+ 0}$, where
    \[
    \begin{aligned}
    \ket{\phi^+} &= \frac{1}{\sqrt{2}}[\ket{00}+\ket{11}],\ket{\phi^-} = \frac{1}{\sqrt{2}}[\ket{00}-\ket{11}],\\
    \ket{\psi^+} &= \frac{1}{\sqrt{2}}[\ket{01}+\ket{10}],\ket{\psi^-} = \frac{1}{\sqrt{2}}[\ket{01}-\ket{10}].
    \end{aligned}
    \]
    We have:
    \[
    \begin{aligned}
    U_3 &= \ket{000}\bra{000} + \ket{001}\bra{001} + \ket{010}\bra{011} + \ket{011}\bra{010}\\
    &+\ket{100}\bra{100} + \ket{101}\bra{101} + \ket{110}\bra{111} + \ket{111}\bra{110}.
    \end{aligned}
    \]
    The Kraus operators of the noise acting on the system, i.e., qubits 1,2 are:
    \[
    S_0 ={}_3\bra{0} U_3 \ket{0}_3 =  \ket{00}\bra{00} + \ket{10}\bra{10},
    \]
    \[
    S_1 ={}_3\bra{1} U_3 \ket{0}_3 = \ket{01}\bra{01} + \ket{11}\bra{11},
    \]
    and those of the noise acting on the environment are:
    \[
    E_0 ={}_{12}\bra{\phi^+} U_3 \ket{\phi^+}_{12} = 0.5[\ket{0}\bra{0} + \ket{0}\bra{1} + \ket{1}\bra{0} + \ket{1}\bra{1}],
    \]
    \[
    E_1 = {}_{12}\bra{\phi^-} U_3 \ket{\phi^+}_{12} = 0.5[\ket{0}\bra{0} - \ket{0}\bra{1} - \ket{1}\bra{0} + \ket{1}\bra{1}] ,
    \]
    \[
    E_2 = {}_{12}\bra{\psi^+} U_3 \ket{\phi^+}_{12} =0,
    \]
    \[
    E_3 = {}_{12}\bra{\psi^-} U_3 \ket{\phi^+}_{12} =0.
    \]
    So, we have
    \[
    S_0^{\dagger}S_0 + S_1^{\dagger}S_1
    = \ket{00}\bra{00} + \ket{01}\bra{01} + \ket{10}\bra{10} + \ket{11}\bra{11} = \mathbb{I},
    \]
    \[
     E_0^{\dagger}E_0 +  E_1^{\dagger}E_1 + E_2^{\dagger}E_2 +  E_3^{\dagger}E_3 
    = \ket{0}\bra{0} + \ket{1}\bra{1} = \mathbb{I}.
    \]
    This implies that both system $S$ and environment $E$ are \textbf{CP-divisible} for GHZ state, and they are unital \cite{self}.
    \item Consider a 3-qubit \textbf{W state} 
    $\frac{1}{\sqrt{3}}(\ket{001} + \ket{010} + \ket{100})$, created from $\ket{100}$, using the unitary:
    \[
    \begin{aligned}
    U &= \ket{000}\bra{000} + \frac{1}{\sqrt{3}}\ket{000}\bra{100} -\frac{1}{\sqrt{3}}\ket{001}\bra{010} \\ 
    &\quad +\frac{1}{\sqrt{3}}\ket{001}\bra{100} - \frac{1}{\sqrt{3}}\ket{010}\bra{001} + \frac{1}{\sqrt{3}}\ket{010}\bra{011}
    \\ 
    &\quad +\frac{1}{\sqrt{3}}\ket{010}\bra{100} + \ket{011}\bra{101} + \frac{1}{\sqrt{3}}\ket{100}\bra{010} \\
    &\quad -\frac{1}{\sqrt{3}}\ket{100}\bra{011} + \frac{1}{\sqrt{3}}\ket{100}\bra{100} + \ket{101}\bra{110} \\
    &\quad +\frac{1}{\sqrt{6}}\ket{110}\bra{001} + \frac{1}{\sqrt{6}}\ket{110}\bra{010} + \frac{1}{\sqrt{6}}\ket{110}\bra{011} \\ 
    &\quad +\frac{1}{\sqrt{2}}\ket{110}\bra{111} + \frac{1}{\sqrt{6}}\ket{111}\bra{001} + \frac{1}{\sqrt{6}}\ket{111}\bra{010}  \\
    &\quad +\frac{1}{\sqrt{6}}\ket{111}\bra{011} - \frac{1}{\sqrt{2}}\ket{111}\bra{111}.
    \end{aligned}
    \]

    Let $U = U_2\cdot U_1$, where $U_1 = iU^2$, $U_2 = -iU^{-1}$. The input state for $U_1$ is $\ket{100}$. We have:
    \[
    \begin{aligned}
    U_1 &= i[\ket{000}\bra{000} + 0.67\ket{001}\bra{001} - 0.67\ket{001}\bra{011} \\
    &\quad +0.33\ket{001}\bra{100} - 0.33\ket{010}\bra{001} + 0.67\ket{010}\bra{010} \\
    &\quad -0.33\ket{010}\bra{011} + 0.57\ket{010}\bra{101} + \ket{011}\bra{110} \\
    &\quad -0.33\ket{100}\bra{001} + 0.33\ket{100}\bra{010} + 0.67\ket{100}\bra{100} \\
    &\quad -0.57\ket{100}\bra{101} + 0.41\ket{101}\bra{001} + 0.41\ket{101}\bra{010} \\
    &\quad +0.41\ket{101}\bra{011} + 0.71\ket{101}\bra{111} + 0.28\ket{110}\bra{001} \\
    &\quad +0.053\ket{110}\bra{010} + 0.52\ket{110}\bra{011} + 0.471\ket{110}\bra{100} \\
    &\quad +0.41\ket{110}\bra{101} - 0.5\ket{110}\bra{111} - 0.28\ket{111}\bra{001} \\
    &\quad -0.52\ket{111}\bra{010} - 0.053\ket{111}\bra{011} + 0.471\ket{111}\bra{100} \\
    &\quad +0.41\ket{111}\bra{101} + 0.5\ket{111}\bra{111}] 
    \end{aligned}
    \]

    The Kraus operators of the noise acting on the system, i.e., qubits 1,2 are:
    \[
    \begin{aligned}
    S_0 &= {}_3\bra{0} U_1 \ket{0}_3 \\
    &= i\ket{00}\bra{00} + 0.67i\ket{01}\bra{01} + 0.33i\ket{10}\bra{01}\\
    &+0.67i\ket{10}\bra{10} + 0.053i\ket{11}\bra{01} + 0.471i\ket{11}\bra{11},
    \end{aligned}
    \]
    \[
    \begin{aligned}
    S_1 &= {}_3\bra{1} U_1 \ket{0}_3 \\
    &= 0.33i\ket{00}\bra{10} + i\ket{01}\bra{11}+ 0.41i\ket{10}\bra{01}\\
    &-0.52i\ket{11}\bra{01} + 0.471i\ket{11}\bra{10},
     \end{aligned}
    \]
    and those of the noise acting on the environment are:
    \[
    E_0 = {}_{12}\bra{00} U_1 \ket{10}_{12} = 0.33i\ket{1}\bra{0} ,
    \]
    \[
    E_1 = {}_{12}\bra{01} U_1 \ket{10}_{12} = 0.57i\ket{0}\bra{1},
    \]
    \[
    E_2 = {}_{12}\bra{10} U_1 \ket{10}_{12} = 0.67i\ket{0}\bra{0} - 0.57i\ket{0}\bra{1},
    \]
    \[
    \begin{aligned}
    E_3 &= {}_{12}\bra{11} U_1 \ket{10}_{12} = 0.471i\ket{0}\bra{0} + 0.41i\ket{0}\bra{1}\\
    &+0.471i\ket{1}\bra{0} + 0.41i\ket{1}\bra{1}.
    \end{aligned}
    \]
    So, we have
    \[
    S_0^{\dagger}S_0 + S_1^{\dagger}S_1
    = \ket{00}\bra{00} +\ket{01}\bra{01} +\ket{10}\bra{10} +\ket{11}\bra{11} = \mathbb{I},
    \]
    \[
     E_0^{\dagger}E_0 +  E_1^{\dagger}E_1 + E_2^{\dagger}E_2 +  E_3^{\dagger}E_3 
    = \ket{0}\bra{0} + \ket{1}\bra{1} = \mathbb{I}.
    \]
    
    Next, the input state for $U_2$ is $0.33i\ket{001} + 0.67i\ket{100} + 0.471i[\ket{110} + \ket{111}]$. We have:
    \[
    \begin{aligned}
    U_2 &= i[-\ket{000}\bra{000} - 0.57\ket{001}\bra{001} + 0.57\ket{001}\bra{010} \\
    &\quad -0.408\ket{001}\bra{110} - 0.408\ket{001}\bra{111} + 0.577\ket{010}\bra{001} \\
    &\quad -0.577\ket{010}\bra{100} - 0.408\ket{010}\bra{110} - 0.408\ket{010}\bra{111} \\
    &\quad -0.577\ket{011}\bra{010} + 0.577\ket{011}\bra{100} - 0.408\ket{011}\bra{110} \\
    &\quad -0.408\ket{011}\bra{111} - 0.577\ket{100}\bra{001} - 0.577\ket{100}\bra{010} \\
    &\quad -0.577\ket{100}\bra{100} - \ket{101}\bra{011} - \ket{110}\bra{101} \\
    &\quad -0.707\ket{111}\bra{110} + 0.707\ket{111}\bra{111}]. 
    \end{aligned}
    \]
    The input effective system state to $U_2$ is: 
    \[
    \begin{aligned}
\vartheta_1\ket{\xi^+}\bra{\xi^+}+\vartheta_2\ket{\xi^-}\bra{\xi^-} +\vartheta_3\ket{\kappa^+}\bra{\kappa^+}
    +\vartheta_4\ket{\kappa^-}\bra{\kappa^-},
    \end{aligned}
    \]
    where
    \[
    \begin{aligned}
    \vartheta_1 &= 0, \vartheta_2= 0.2209, \vartheta_3 = 0.7791, \vartheta_4 = 0,\\
    \ket{\xi^+}&= 0.7595\ket{00} - 0.6291\ket{01} + 0.1653\ket{10} + 0\ket{11}, \\
    \ket{\xi^-}&= 0\ket{00} + 0\ket{01} + 0\ket{10} + 1\ket{11}, \\
    \ket{\kappa^+}&= 0.3741\ket{00} + 0.6304\ket{01} + 0.6802\ket{10} + 0\ket{11}, \\
    \ket{\kappa^-}&= -0.5321\ket{00} - 0.4548\ket{01} + 0.7142\ket{10} + 0\ket{11}.
    \end{aligned}
    \]
    The input effective environment state to $U_2$ is: 
    \[
    \begin{aligned}
    \varsigma_1\ket{\chi^+}\bra{\chi^+}+\varsigma_2\ket{\chi^-}\bra{\chi^-},
    \end{aligned}
    \]
    where
    \[
    \begin{aligned}
    \varsigma_1 &= 0.7791, \varsigma_2 = 0.2209, \\
    \ket{\chi^+}&= 0.8967\ket{0} - 0.4426\ket{1}, \\
    \ket{\chi^-}&= 0.4426\ket{0} + 0.8967\ket{1}
    \end{aligned}
    \]
    
    The Kraus operators of the noise acting on the system, i.e., qubits 1,2 are:
    \[
    \begin{aligned}
    S_0 &= \sqrt{\varsigma_1}\bra{\chi^+} U_2 \ket{\chi^+} +\sqrt{\varsigma_2}\bra{\chi^+} U_2  \ket{\chi^-} \\
    &= -0.73i\ket{00}\bra{00} + 0.149i\ket{00}\bra{01} - 
    0.252i\ket{00}\bra{11}\\ 
    &+0.42i\ket{01}\bra{00}
    - 0.149i\ket{01}\bra{01} - 0.153i\ket{01}\bra{10}\\ 
    &-0.762i\ket{01}\bra{11} - 0.42i\ket{10}\bra{00} - 0.661i\ket{10}\bra{01}\\
    &-0.302i\ket{10}\bra{10} - 0.728i\ket{11}\bra{10} + 0.0715i\ket{11}\bra{11},
    \end{aligned}
    \]
    \[
    \begin{aligned}
    S_1 &= \sqrt{\varsigma_1}\bra{\chi^-} U_2 \ket{\chi^+} +\sqrt{\varsigma_2}\bra{\chi^-} U_2 \ket{\chi^-} \\
    &= -0.162i\ket{00}\bra{00} + 0.302i\ket{00}\bra{01} - 0.51i\ket{00}\bra{11}\\
    &-0.207i\ket{01}\bra{00}
    - 0.302i\ket{01}\bra{01} + 0.451i\ket{01}\bra{10} \\
    &- 0.258i\ket{01}\bra{11} + 0.207i\ket{10}\bra{00} - 0.579i\ket{10}\bra{01}\\
     &+ 0.149i\ket{10}\bra{10} + 0.359i\ket{11}\bra{10} + 0.145i\ket{11}\bra{11},
     \end{aligned}
    \]
    and those of the noise acting on the environment are:
    \[
    \begin{aligned}
    E_0 &=\sqrt{\vartheta_1}\bra{\xi^+} U_2 \ket{\xi^+} +\sqrt{\vartheta_2}\bra{\xi^+} U_2 \ket{\xi^-} \\ 
    &+ \sqrt{\vartheta_3}\bra{\xi^+} U_2 \ket{\kappa^+} + \sqrt{\vartheta_4}\bra{\xi^+}U_2\ket{\kappa^-}\\
    &= -0.079i\ket{0}\bra{0} + 0.509i\ket{0}\bra{1}\\
    &+0.027i\ket{1}\bra{0} - 0.22i\ket{1}\bra{1},
    \end{aligned}
    \]
    \[
    \begin{aligned}
    E_1 &= \sqrt{\vartheta_1}\bra{\xi^-} U_2 \ket{\xi^+} +\sqrt{\vartheta_2}\bra{\xi^-} U_2 \ket{\xi^-}\\ 
    &+ \sqrt{\vartheta_3}\bra{\xi^-} U_2 \ket{\kappa^+} + \sqrt{\vartheta_4}\bra{\xi^-}U_2\ket{\kappa^-}\\
    &= -0.42i\ket{0}\bra{0} + 0.462i\ket{0}\bra{1}\\
    &+0.24i\ket{1}\bra{0} - 0.081i\ket{1}\bra{1},
    \end{aligned}
    \]
    \[
    \begin{aligned}
    E_2  &= \sqrt{\vartheta_1}\bra{\kappa^+} U_2 \ket{\xi^+} +\sqrt{\vartheta_2}\bra{\kappa^+} U_2 \ket{\xi^-} \\
    &+\sqrt{\vartheta_3}\bra{\kappa^+} U_2 \ket{\kappa^+} + \sqrt{\vartheta_4}\bra{\kappa^+}U_2\ket{\kappa^-}\\
    &= -0.327i\ket{0}\bra{0} - 0.581i\ket{0}\bra{1} \\
    &-0.238i\ket{1}\bra{0} + 0.196i\ket{1}\bra{1},
    \end{aligned}
    \]
    \[
   \begin{aligned}
    E_3 &= \sqrt{\vartheta_1}\bra{\kappa^-} U_2 \ket{\xi^+} +\sqrt{\vartheta_2}\bra{\kappa^-} U_2 \ket{\xi^-}  \\
    &+\sqrt{\vartheta_3}\bra{\kappa^-} U_2 \ket{\kappa^+} + \sqrt{\vartheta_4}\bra{\kappa^-}U_2\ket{\kappa^-}\\
    &= -0.687i\ket{0}\bra{0} - 0.256i\ket{0}\bra{1} \\
    &+
    0.347i\ket{1}\bra{0} - 0.17i\ket{1}\bra{1}.
    \end{aligned}
    \]
So, we have
    \[
    S_0^{\dagger}S_0 + S_1^{\dagger}S_1
    = \ket{00}\bra{00} + \ket{01}\bra{01} + \ket{10}\bra{10} + \ket{11}\bra{11} = \mathbb{I},
    \]
    \[
     E_0^{\dagger}E_0 +  E_1^{\dagger}E_1 + E_2^{\dagger}E_2 +  E_3^{\dagger}E_3 
    = \ket{0}\bra{0} + \ket{1}\bra{1} = \mathbb{I}.
    \]
   
    This implies that both system $S$ and environment $E$ are \textbf{CP-divisible} for W state, despite being non-unital \cite{self}.

    \item Consider the single-qubit transpose operation in $\mathcal{B}$ form \cite{Jagadish_2018}, $\mathcal{B}_T=\left[\begin{array}{cccc}
        1 & 0 & 0 & 0 \\
        0 & 0 & 1 & 0 \\
        0 & 1 & 0 & 0 \\
        0 & 0 & 0 & 1
    \end{array}\right]$. Clearly, this is a non-CP map, since the Kraus operators are $E_0:= D_0=\ket{0}\bra{0}$, $E_1:=D_1 = \ket{1}\bra{1}$, $E_2:=D_2 = \ket{0}\bra{1}+\ket{1}\bra{0}$ and $E_3:=F_3 = \ket{1}\bra{0}-\ket{0}\bra{1}$, which satisy the (trace-preserving) condition $D_0^\dagger D_0+D_1^\dagger D_1+D_2^\dagger D_2-F_3^\dagger F_3=\mathbb{I}$, but not the completeness relation $\sum_{i=0}^3E_i^\dagger E_i=\mathbb{I}$.%, and 
    %the channel is non-unital, i.e.~$\sum_{i=0}^3E_iE_i^\dagger\neq\mathbb{I}$. 
    Then, the $\mathcal{A}$ form \cite{Jagadish_2018} of the partial transpose operation on a two-qubit system, obtained from the matrix $\mathbb{I}\otimes\mathcal{B}_T$:
    
    \relscale{0.54}
    \[
    \begin{aligned}
        \left[\begin{array}{cccccccccccccccc}
            1 & 0 & 0 & 0 & 0 & 0 & 0 & 0 & 0 & 0 & 1 & 0 & 0 & 0 & 0 & 0 \\
            0 & 0 & 0 & 0 & 0 & 0 & 0 & 0 & 0 & 0 & 0 & 0 & 0 & 0 & 0 & 0 \\
            0 & 1 & 0 & 0 & 0 & 0 & 0 & 0 & 0 & 0 & 0 & 1 & 0 & 0 & 0 & 0 \\
            0 & 0 & 0 & 0 & 0 & 0 & 0 & 0 & 0 & 0 & 0 & 0 & 0 & 0 & 0 & 0 \\
            0 & 0 & 0 & 0 & 1 & 0 & 0 & 0 & 0 & 0 & 0 & 0 & 0 & 0 & 1 & 0 \\
            0 & 0 & 0 & 0 & 0 & 0 & 0 & 0 & 0 & 0 & 0 & 0 & 0 & 0 & 0 & 0 \\
            0 & 0 & 0 & 0 & 0 & 1 & 0 & 0 & 0 & 0 & 0 & 0 & 0 & 0 & 0 & 1 \\
            0 & 0 & 0 & 0 & 0 & 0 & 0 & 0 & 0 & 0 & 0 & 0 & 0 & 0 & 0 & 0 \\
            0 & 0 & 0 & 0 & 0 & 0 & 0 & 0 & 0 & 0 & 0 & 0 & 0 & 0 & 0 & 0 \\
            1 & 0 & 0 & 0 & 0 & 0 & 0 & 0 & 0 & 0 & 1 & 0 & 0 & 0 & 0 & 0 \\
            0 & 0 & 0 & 0 & 0 & 0 & 0 & 0 & 0 & 0 & 0 & 0 & 0 & 0 & 0 & 0 \\
            0 & 1 & 0 & 0 & 0 & 0 & 0 & 0 & 0 & 0 & 0 & 1 & 0 & 0 & 0 & 0 \\
            0 & 0 & 0 & 0 & 0 & 0 & 0 & 0 & 0 & 0 & 0 & 0 & 0 & 0 & 0 & 0 \\
            0 & 0 & 0 & 0 & 1 & 0 & 0 & 0 & 0 & 0 & 0 & 0 & 0 & 0 & 1 & 0 \\
            0 & 0 & 0 & 0 & 0 & 0 & 0 & 0 & 0 & 0 & 0 & 0 & 0 & 0 & 0 & 0 \\
            0 & 0 & 0 & 0 & 0 & 1 & 0 & 0 & 0 & 0 & 0 & 0 & 0 & 0 & 0 & 1
        \end{array}\right],
    \end{aligned}
    \]\normalsize
    acting on the vectorized pure Bell state density matrix, \relscale{0.55}${ \left[\begin{array}{cccccccccccccccc} \frac{1}{2} & 0 & 0 & \frac{1}{2} & 0 & 0 & 0 & 0 & 0 & 0 & 0 & 0 & \frac{1}{2} & 0 & 0 & \frac{1}{2} \end{array}\right]^T}$\normalsize 
    
    yields \relscale{0.55}$\left[\begin{array}{cccccccccccccc} \frac{1}{2} & 0 & 0 & 0 & 0 & 0 & \frac{1}{2} & 0 & 0 & \frac{1}{2} & 0 & 0 & 0 & 0\end{array}\right.$\\ $\left.\begin{array}{cc}  0 & \frac{1}{2} \end{array}\right]^T$\normalsize,\, that upon matricizing and realignment gives a mixed state (pseudo-) density matrix \cite{pseudo} $\left[\begin{array}{cccc} \frac{1}{2} & 0 & 0 & 0 \\ 0 & 0 & \frac{1}{2} & 0 \\ 0 & \frac{1}{2} & 0 & 0 \\ 0 & 0 & 0 & \frac{1}{2} \end{array}\right]$, that is non-positive-semi-definite, as expected as a signature of the Bell state being entangled \cite{peres}. This implies that $\mathbb{I}$ being a CP map cannot be the system channel, if the environment channel is a non-CP partial transpose $\mathcal{B}_T$ map. If we instead consider $\mathcal{B}_T\otimes\mathcal{B}_T$, the input pure Bell state remains unchanged at the output, suggesting that $\mathcal{B}_T\otimes\mathcal{B}_T$ is a unitary (noiseless) CPTP map, with both the system and environment being individually non-CP.
\end{enumerate}
%\section{Conclusion}\label{sec:conclude}

\bibliographystyle{IEEEtran}
\bibliography{cptp}
\end{document}